\documentclass[
BCOR2pt, 
captions=nooneline, 
bibliography=totoc, 
numbers=noenddot, 
parskip=half, 
headings=normal, 
abstracton 
]{scrartcl} 

\usepackage[USenglish]{babel} 
\usepackage{graphicx} 
\usepackage{remreset} 
\usepackage[nouppercase]{scrpage2} 
\usepackage{amsmath} 
\usepackage{amssymb} 
\usepackage[thmmarks, 
amsmath, 
hyperref 
]{ntheorem} 
\usepackage{bm} 
\usepackage{bbm} 
\usepackage{enumitem} 
\usepackage[hang]{subfigure} 
\usepackage{wrapfig} 
\usepackage{tabularx} 
\usepackage{dcolumn} 
\usepackage{booktabs} 
\usepackage{listings} 
\usepackage{psfrag} 
\usepackage[pdftex, 
setpagesize=false, 
pdfborder={0 0 0}, 
pdfpagemode=UseOutlines, 
]{hyperref} 

\makeatletter
\newcommand\mytoday{\number\year-\ifcase\month\or 01\or 02\or 03\or 04\or 05\or 06\or 07\or 08\or 09\or 10\or 11\or 12\fi-\ifcase\day\or 01\or 02\or 03\or 04\or 05\or 06\or 07\or 08\or 09\or 10\or 11\or 12\or 13\or 14\or 15\or 16\or 17\or 18\or 19\or 20\or 21\or 22\or 23\or 24\or 25\or 26\or 27\or 28\or 29\or 30\or 31\fi} 
\makeatother
\setkomafont{sectioning}{\normalcolor\bfseries} 
\clearscrheadfoot 
\automark[section]{subsection} 
\setkomafont{captionlabel}{\bfseries} 
\newcolumntype{d}[2]{D{.}{.}{#1.#2}} 
\newcommand*{\abstractnoindent}{} 
\let\abstractnoindent\abstract
\renewcommand*{\abstract}{\let\quotation\quote\let\endquotation\endquote
\abstractnoindent}
\makeatletter
\renewcommand{\p@enumii}[1]{\theenumi(#1)}
\makeatother


\theoremstyle{break} 
\theoremheaderfont{\bfseries}
\theorembodyfont{}
\theoremseparator{}
\newtheorem{definition}{Definition}[section] 

\newtheorem{theorem}[definition]{Theorem}
\newtheorem{proposition}[definition]{Proposition}

\newtheorem{remark}[definition]{Remark}

\newtheorem{algorithm}{Algorithm}
\theoremstyle{nonumberbreak} 
\theoremsymbol{$\Box$}
\newtheorem{proof}{Proof}

\newcommand*{\IP}{\mathbb{P}}
\newcommand*{\IE}{\mathbb{E}}
\newcommand*{\IR}{\mathbb{R}}

\newcommand*{\IN}{\mathbb{N}}
\newcommand*{\I}[1]{\textbf{1}_{#1}}

\newcommand*{\F}{\mathcal{F}}

\begin{document}
\thispagestyle{plain}
	\begin{center}
	{\bfseries\Large Consistent iterated simulation of multi-variate default times: a Markovian indicators characterization}

		\par\bigskip
		\vspace{0cm}
		{\Large Damiano Brigo}\\\vspace{0.2cm}
			{Dept of Mathematics, Imperial College London}\\
		{180 Queen's Gate, London}\\
		{email: damiano.brigo@imperial.ac.uk}\\		
    \vspace{0.5cm}
		{\Large Jan-Frederik Mai}\\
		\vspace{0.2cm}
		{XAIA Investment GmbH,}\\
		{Sonnenstr.\ 19, 80331 M\"unchen, Germany,}\\
		{email: jan-frederik.mai@xaia.com,}\ \ {phone: +49 89 589275-131}\\
    \vspace{0.5cm}
    {\Large Matthias Scherer}\\
		\vspace{0.2cm}
		{Chair of Mathematical Finance,}\\
		{Technische Universit{\"a}t M{\"u}nchen,}\\
		{Parkring 11, 85748 Garching-Hochbr\"uck, Germany,}\\
		{email: scherer@tum.de,}\ \  {phone: +49 89 289-17402} \\ \hspace{1cm} \\ \hspace{1cm} \\ First Version:  June 7, 2013. This Version: \today
	\end{center}

\begin{abstract}
We investigate under which conditions a single simulation of joint default times at a final time horizon can be decomposed into a set of simulations of joint defaults on subsequent adjacent sub-periods leading to that final horizon. Besides the theoretical interest, this is also a practical problem as part of the industry has been working under the misleading assumption that the two approaches are equivalent for practical purposes. As a reasonable trade-off between realistic stylized facts, practical demands, and mathematical tractability, we propose models leading to a Markovian multi-variate survival--indicator process, and we investigate two instances of static models for the vector of default times from the statistical literature that fall into this class. On the one hand, the ``looping default'' case is known to be equipped with this property at least since \cite{herbertsson08,bielecki11}, and we point out that it coincides with the classical ``Freund distribution'' in the bivariate case. On the other hand, if all sub-vectors of the survival indicator process are Markovian, this constitutes a new characterization of the Marshall--Olkin distribution, and hence of multi-variate lack-of-memory. A paramount property of the resulting model is stability of the type of multi-variate distribution with respect to elimination or insertion of a new marginal component with marginal distribution from the same family.  The practical implications of this ``nested margining'' property are enormous. 
To implement this distribution we present an efficient and unbiased simulation algorithm based on the L\'evy-frailty construction. We highlight different pitfalls in the simulation of dependent default times and examine, within a numerical case study, the effect of inadequate simulation practices.
\end{abstract}


Classification Codes: {\bf AMS}  60E07, 62H05, 62H20, 62H99;  \ 
{\bf JEL} C15, C16.

\noindent {Keywords:} \ \ Stepwise default simulation, default modeling, credit modeling, default dependence, default correlation, default simulation, arrival times, credit risk, Marshall--Olkin distribution, nested margining, Freund distribution, looping default models.

\section{Introduction}
The increasingly global nature of financial products and risks is calling for adequately complex stochastic models and simulation procedures, often involving thousands of risk factors that can be different in nature. This is required for valuation purposes and for risk measurement. Investment banks and financial services companies are devoting a sizable effort to designing software and hardware architectures that support such global simulations effectively, see, e.g. \cite{albanese2011}. The path-dependent nature of many risks and the necessity to analyze risks at different time horizons lead to an iterated simulation of all risk factors across time steps. The way to consistently represent statistical dependence for default times in each single step of the simulation is the main motivation of this paper. When simulating default times over a final finite horizon two approaches are possible, broadly speaking: 
\begin{itemize}
\item[(i)] Simulate each default time once and for all in each given scenario, and store its value while the other risk factors are iterated through time in that scenario up to the final simulation horizon, so as to properly account for default when the time comes in that specific path. 
\item[(ii)] Alternatively, one may decide simply to simulate a ``default/no default'' indicator at each time step of the common iteration for all risk factors, the indicator flaring up in the specific step, before the final horizon, where default occurs, or not flaring at all if no default occurs prior to the final horizon. 
\end{itemize}
On the one hand, the mathematical underpinning -- if any -- for company-wide, global simulations of default is often a static model for the vector of default times, usually a copula-based ansatz. Such a model originates from the statistical literature and renders the approach (i) more natural. On the other hand, when dealing with large portfolios, the literature in financial risk management mostly prefers models relying on a repeated evolution of risk factors on common time grids.  The approach (ii) is more consistent with this way of thinking and therefore more desirable both from a theoretical and practical point of view when dealing with large portfolios, as we will further explain below. Hence the choice between (i) and (ii) originates the question: What are convenient conditions on the multi-variate distribution of the default times such that the two approaches above are consistent? The present article dares an important step towards bridging the gap between the static modeling of default times and the dynamic modeling of default indicators, deriving en passant new economic interpretations for two families of multivariate distributions that have been considered under a different focus in the statistical and applied literature in the past. We are going to stress that for most stochastic models the approach (ii), while aligning default simulation with other risk factors simulations along time, is far more complicated than (i), because conditional multi-variate survival probabilities are complicated objects in general. Finding statistical models for the default times that allow for a convenient implementation of (ii) is related to a multi-variate notion of lack-of-memory and is important for multiple reasons:
\begin{itemize}
\item \textbf{Software consistency with ``Brownian-driven'' asset classes}: Consider a bank that runs a global simulation on a large portfolio, including complex products and defaults, in order to obtain a risk measure. One example would be computing the value at risk or the expected shortfall of CVA, a task that is numerically very intensive, see, e.g. \cite{BrigoPallaMorini}. In this context, there is need to evolve risk factors according to controlled time steps that are common to all factors, to have all required variables at each step of the simulation. While this is relatively natural for asset models that are driven by Brownian type processes and even extensions with jumps, it becomes harder when trying to include default of underlying entities or counterparties. The reason for this is that default times, typically represented through intensity models, should be simulated just once, being static random variables as opposed to random processes. Once simulated, there would be nothing left to iterate. However, the consistency of the global simulation and the desire to have all variables simulated at every step is prompting the design of iterated survival or default flags across the time steps that are already used in the simulation of more traditional assets. 
\item \textbf{Basel III requirement for risk horizons:} A further motivation for iterating the global simulation across standard time steps is coming from the Basel III framework \cite{baselannex4} when trying to address liquidity risk. BIS suggests the following solution:\par \emph{``The Committee has agreed that the differentiation of market liquidity across the trading book will be based on the concept of liquidity horizons. It proposes that banks' trading book exposures be assigned to a small number of liquidity horizon categories. [10 days, 1 month, 3 months, 6 months, 1 year]. The shortest liquidity horizon (most liquid exposures) is in line with the current 10-day VaR treatment in the trading book. The longest liquidity horizon (least liquid exposures) matches the banking book horizon at one year. The Committee believes that such a framework will deliver a more graduated treatment of risks across the balance sheet.
Among other benefits, this should also serve to reduce arbitrage opportunities between the banking and trading books.''}
\par
It is clear then that a bank will need to simulate the risk factors of the portfolio across a grid including the standardized holding periods above. In this sense it will be practical to simulate all variables, including defaults and survivals, in the common time steps. Software architecture and the possibility to effectively decompose the simulation across steps, prompt to the possibility to iterate the default simulation rather than trying to simulate random default times just once.
\item \textbf{Rectifying existing market practice:} Part of the industry has been iterating dependence structures of static multi-variate default times across common time steps. While for single exponentially distributed random variables the lack-of-memory property allows to do so, for the dependence structure to be iterated one needs a meaningful multi-variate extension of the lack-of-memory property. This problem has been addressed initially in \cite{BrigoChourdakisSelfChaining}, who provide conditions for consistency of the two approaches when the grid is the same for all risk factors, but only in a partial way (as we will explain more in detail below). 
\item \textbf{General need for dependence modeling in the context of the current counterparty credit risk debate:} As an example, the current debate on valuation adjustments (as the partly overlapping credit CVA, debit DVA, and funding ``FVA'' adjustments, see e.g. \cite{BrigoPallaMorini}), is forcing financial institutions to run global simulations over very large portfolios. By nature, CVA is an option on a very large portfolio containing the most disparate risk factors. A key quantity in valuing this option is the dependence between the default of a counterparty and the value of the underlying portfolio that is traded with that counterparty. When such dependence takes its worst possible value for the agent making the calculation we have wrong way risk (WWR), a risk that is at the centre of the agenda of the Bank of International Settlements (BIS) in reforming current regulation. 
Modeling the dynamics of dependence is not only essential for the current emergencies of the industry, such as CVA/DVA/FVA and risk measures on these quantities, but it is also necessary for the management of pure credit products, such as, e.g., Collateralized Debt Obligations (CDO).
\end{itemize}
We aim at increasing awareness of the fact that the stepwise simulation of default indicators ((ii) above) is a hard task in general, and in particular that the practical implementation is not feasible without huge efforts (both theoretical and computational). 

\medskip

The  paper is organized as follows.
\par
To be able to technically discuss the single step versus multi step simulation of default times, we study in Section \ref{sec:markovind} the Markovianity for a vector of survival indicators. We explain that if this property is not there, the simulation is very difficult in general and infeasible for global simulations on a company-wide level. We then adopt the Markov assumption. This solves a number of problems and is still flexible enough to include the looping default model, as has already been shown in \cite{herbertsson08,bielecki11}. We bridge the gap to classic literature by pointing out that this leads to the Freund distribution for the bivariate default time statistics, see \cite{freund61}. Multivariate extensions in turn lead to easy simulation through matrix exponentials.
\par
In Section \ref{sec:newmarshallolkin} we illustrate how even the solution based on a  Markovian vector of survival indicators retains some problems.
In particular, portfolio re-balancing issues and lack of nested marginalization are undesired properties. We then show that problems of the Markovian version are solved if we also request that all sub-vectors of indicators are Markovian. This leads to the main inner-mathematical innovation of the paper: the Markov property for sub-vectors is  equivalent to have a Marshall--Olkin distribution for the multi-variate default times.
We provide an unbiased simulation scheme for Marshall--Olkin distributions, and discuss efficient Marshall--Olkin parameterizations. In particular, we review the L{\'e}vy-frailty model, possibly with factor structures, which is perfectly suited for global simulations.
\par
In Section \ref{sec:casestudy} we return to the original question of this introduction, namely the consistency of approaches (i) and (ii) above for simulation, and re-discuss the paper\linebreak \cite{BrigoChourdakisSelfChaining}, showing why its analysis is only partial as it assumes homogeneous time steps and further it focuses on the univariate indicator of a vector of survivals rather than on the vector of survival indicators. We see how splitting marginal distributions  and dependence structure is not always a good idea. The further point is made that if one simulates different default times with different time steps, then even self-chaining/extreme value copulas advocated in \cite{BrigoChourdakisSelfChaining} cannot be iterated and the only possibility is the Marshall--Olkin distribution.
\par
The final section concludes the paper. 
\section{Markovian survival indicator processes}\label{sec:markovind}
We consider a random vector of default times $(\tau_1,\ldots,\tau_d)$ and its associated indicator process $\textbf{Z}(t)=(\I{\{\tau_1>t\}},\ldots,\I{\{\tau_d>t\}})$, formally defined on $(\Omega,\F,\IP)$. We face the task of simulating a path of $\textbf{Z}$ along an equidistant grid with length $\Delta$, i.e.\ the sequence $(\textbf{Z}(0), \textbf{Z}(\Delta), \textbf{Z}(2\,\Delta),\ldots)$. In the sequel it will be convenient to identify the state space $\{0,1\}^d$ of $\textbf{Z}$ with the power set of $\{1,\ldots,d\}$ via the bijection 
\begin{gather}
h(I):=(\I{\{1 \in I\}},\ldots,\I{\{d \in I\}}), \quad I \subset \{1,\ldots,d\}.
\label{bijectionh}
\end{gather}
In order to carry out the simulation in a stepwise manner, in step $k$ of the simulation we have to simulate $\textbf{Z}(k\,\Delta)$ from the discrete distribution 
\begin{gather*}
\Big( P_{\textbf{Z}((k-1)\,\Delta), h(J)}\big[{(k-1)},{\F_{(k-1)\,\Delta}}\big] \Big)_{J \subset \{1,\ldots,d\}},
\end{gather*}
where 
\begin{gather*}
P_{h(I),h(J)}\big[{(k-1)},{\F_{(k-1)\,\Delta}}\big]:=\IP\big(\textbf{Z}(k\,\Delta)=h(J)\,\big|\,(\textbf{Z}((k-1)\,\Delta)=h(I),\F_{(k-1)\,\Delta} \big),
\end{gather*}  
with $\F_t$ being the $\sigma$-algebra of all available information at time $t$. In the sequel we list several issues demonstrating why this procedure is a very hard task. In general, the transition probabilities depend on the $\sigma$-algebra $\F_t$ generated by a battery of risk factors. This causes the following problems:
\begin{itemize}
\item[(a)] In reality, default risk is correlated with risk factors of other asset classes such as, for example, equity derivatives. The development of such a global model requires huge efforts and is therefore typically not implemented in practice. In particular, such a design requires different departments of a financial institution to work together, which might be infeasible. It is common to split the business into several sections and every section models their specific risk factors with an appropriate level of sophistication. Typically, these levels do not have a common denominator. For instance, it is likely that a swap desk uses a stochastic interest rate model, whereas a credit desk uses deterministic interest rates and focuses on the stochastic evolution of credit spreads instead. On a global level, these two approaches are inconsistent of course.
\item[(b)] The transition probabilities might not be easy to compute. Typically, there do not exist closed form expressions for them, and numerical integration techniques - if available at all - become time-consuming and difficult to implement.
\item[(c)] If the transition probabilities depend on the whole histories of certain risk factors, then these paths have to be stored, leading to a critical algorithm, especially for large dimensions. This already applies if $\F_t$ only depends on the history of $\textbf{Z}$, e.g.\ in case the timing of previous defaults effects future defaults.  
\end{itemize}
Furthermore, even if we drop the dependence on $\F_t$, the dependence of the transition probabilities on the time step $k$ still might cause serious practical problems:
\begin{itemize}
\item[(d)] Overparameterization: For each time step $k$ we have to deal with a whole matrix of transition probabilities. Especially for thin grids and large portfolios, this becomes a challenging issue.
\item[(e)] The number of time steps, and hence the number of parameters, depends on the grid length $\Delta$. In case we need to run simulations for several different $\Delta$ (e.g., daily, weekly, monthly according to Basel III requirements), we have to re-design the algorithm each time, because it is not $\Delta$-independent.
\end{itemize}
To circumvent these difficulties, a convenient trade-off between realism and tractability is the assumption of $\textbf{Z}$ being a continuous-time Markov chain, i.e.\ a time-homogeneous Markov process. In particular, this implies that
\begin{gather*}
P_{h(I),h(J)}\big[{(k-1)},{\F_{(k-1)\,\Delta}}\big] = P_{h(I),h(J)}\big[\Delta\big],
\end{gather*}
depends only on $\Delta$. Indeed, this restriction resolves all issues (a)--(e). On a first glimpse, this assumption appears to be restrictive. For instance, it implies that we choose $(\tau_1,\ldots,\tau_d)$ from a certain finite parametric family of distributions, since continuous-time Markov chains on a $2^d$-dimensional state space are determined by a $2^d \times 2^d$-dimensional intensity matrix. 

However, the earlier references \cite{herbertsson08,bielecki11} study models of this type in quite some detail and show how to make them workable in the context of credit derivative pricing. The focus of part of the present article is to bridge the gap between the statistical literature, which defines $(\tau_1,\ldots,\tau_d)$ by a static model, and the dynamic Markov chain setup introduced above, which has been investigated in \cite{herbertsson08,bielecki11}.
\subsection{Reconsidering the looping default model}
One of the most intuitive models for contagion effects in portfolio credit risk is the so-called ``looping default'', the terminology being introduced in one of the first works on counterparty credit risk pricing by \cite{jarrow01}. The intuition of this model is that companies have an initially constant hazard rate, but a default event of one company changes hazard rates of the surviving companies. Despite the looping default model being an intuitively reasonable approach, it turns out that constructing a well-defined probability space supporting such a multi-variate distribution is surprisingly difficult. When writing down the canonical construction of default times in classical intensity-based models there is a recursive dependence of one default time on the other default times. \cite{jarrow01} resolve this issue by simplifying the model to a case when the involved companies are split into two groups and only the defaults of group $\mathcal{A}$ cause changes of the hazard rates in group $\mathcal{B}$, but not vice versa, which is no longer a real looping default model. However, the problem has been investigated further in subsequent articles and finally was resolved by \cite{yu07} who constructs the looping default using the so-called ``total hazard construction'', which originates from the statistical literature, see \cite{norros86,shaked87}. The total hazard construction defines a $d$-dimensional random vector of default times as a function of $d$ independent random variables, such that the corresponding default intensities satisfy certain relations that are specified a priori. However, this construction algorithm is rather complicated to implement in practice, and in particular has no natural coherence with stepwise simulation - rendering it inconvenient for our purpose. As a first example of the total hazard construction, \cite{yu07} reconsiders the looping default of \cite{jarrow01} in a two-dimensional setup. The references \cite{herbertsson08,bielecki11} show that the looping default model falls into the class of default models whose survival indicator process is a Markov chain, which provides an alternative stochastic construction being naturally consistent with stepwise simulation. Interestingly, in the bivariate case the probability law of $(\tau_1,\tau_2)$ is well-known in the statistical literature as well. In this respect, we point out the following 
\begin{remark}[The looping default model and the Freund distribution] The bivariate distribution which is derived in \cite{yu07} coincides precisely with the so-called bivariate Freund distribution, which is an ``old friend'' from reliability theory, see \cite{freund61}. In other words, the looping default has incidentally been known for many years in the statistical literature by the name ``Freund distribution''. The fact that both distributions coincide can be observed by comparing the bivariate densities derived in \cite{yu07} and \cite{freund61}, respectively. We will provide details below. 
\end{remark}
In the sequel, we provide a new construction for the Freund distribution based on Markov chains, which in our view provides a simpler access to this probability law, and in particular can be simulated stepwise in a very easy way. Moreover, it can easily be lifted to dimensions $d>2$ and to extensions with joint defaults.
\par
We consider two companies' default times $(\tau_1,\tau_2)$. Let $\lambda_1,\lambda_2,\tilde{\lambda}_1,\tilde{\lambda}_2>0$ be model parameters satisfying the constraint $\tilde{\lambda}_i \neq \lambda_1+\lambda_2$, $i=1,2$. We construct the associated survival indicator process $\textbf{Z}(t):=(\I{\{\tau_1>t\}},\I{\{\tau_2>t\}})$ as a time-homogeneous continuous-time Markov chain. This process is fully described by its so-called intensity matrix $Q$, which algebraically is a $4 \times 4$-matrix with vanishing row sums and all off-diagonal elements being non-negative. Indexing the four states $(1,1),(0,1),(1,0),(0,0)$ by the numbers $1,2,3,4$ we define the $Q$-matrix as
\begin{gather*}
Q = \begin{pmatrix}
-(\lambda_1+\lambda_2) & \lambda_1 & \lambda_2 & 0 \\
0 & -\tilde{\lambda}_2 & 0 & \tilde{\lambda}_2 \\
0 & 0 & -\tilde{\lambda}_1 & \tilde{\lambda}_1 \\
0 & 0 & 0 & 0 
\end{pmatrix}
\end{gather*}
This matrix has to be read as follows: Being in a certain state corresponds to being in a certain row of the matrix. For instance, the process starts in state $(1,1)$ corresponding to row $1$. Now for each other state $(0,1),(1,0),(0,0)$ there is a latent exponential random variable, which describes the time span before the chain moves there. The exponential rate of the corresponding random variables are given as the entries $Q_{(1,1),(0,1)},Q_{(1,1),(1,0)},Q_{(1,1),(0,0)}$, i.e.\ in columns $2,3,4$ of row $1$, respectively. For instance, the chain cannot go directly from zero default $(1,1)$ to joint default $(0,0)$, hence the respective rate equals $Q_{(1,1),(0,0)}=0$. However, the first company has hazard rate $\lambda_1$ and the second has hazard rate $\lambda_2$, determining the entries $Q_{(1,1),(0,1)}$ and $Q_{(1,1),(1,0)}$. The diagonal entry $Q_{(1,1),(1,1)}$ finally is the negative of the sum over all other entries in the row, stochastically being the negative of the exponential rate of the occupation time in state $(1,1)$. This is because the minimum of independent exponential random variables is again exponential, and the respective exponential rates add up. The same logic applies to the other rows of $Q$. In particular, after the default of one company, the hazard rate of the remaining company changes from $\lambda_i$ to $\tilde{\lambda}_i$, and the bottom row of $Q$ is zero because both companies being bankrupt is an absorbing state. Using diagonalization, one can show that
\begin{gather*}
{P}[t]:=e^{t\,Q}\in\IR^{4\times 4},\quad t \geq 0,
\end{gather*}    
is given by
\begin{align*}
P_{(1,1),(1,1)}[t] & = e^{-(\lambda_1+\lambda_2)\,t},\\
P_{(1,1),(0,1)}[t] & = \frac{\lambda_1}{\lambda_1+\lambda_2-\tilde{\lambda}_2}\,\Big(e^{-\tilde{\lambda}_2\,t}-e^{-(\lambda_1+\lambda_2)\,t}\Big),\\
P_{(1,1),(1,0)}[t] & = \frac{\lambda_2}{\lambda_1+\lambda_2-\tilde{\lambda}_1}\,\Big(e^{-\tilde{\lambda}_1\,t}-e^{-(\lambda_1+\lambda_2)\,t}\Big),\\
P_{(1,1),(0,0)}[t] & = -\frac{\lambda_1}{\lambda_1+\lambda_2-\tilde{\lambda}_2}\,e^{-\tilde{\lambda}_2\,t}-\frac{\lambda_2}{\lambda_1+\lambda_2-\tilde{\lambda}_1}\,e^{-\tilde{\lambda}_1\,t}\\
& \quad +1+\Big(\frac{\lambda_1}{\lambda_1+\lambda_2-\tilde{\lambda}_2}+\frac{\lambda_2}{\lambda_1+\lambda_2-\tilde{\lambda}_1}-1 \Big)\,e^{-(\lambda_1+\lambda_2)\,t},\\
P_{(0,1),(0,1)}[t] & = e^{-\tilde{\lambda}_2\,t},\quad P_{(0,1),(0,0)}(t)  = 1-e^{-\tilde{\lambda}_2\,t},\\
P_{(1,0),(1,0)}[t] & = e^{-\tilde{\lambda}_1\,t},\quad P_{(1,0),(0,0)}(t)  = 1-e^{-\tilde{\lambda}_1\,t},
\end{align*}
and all other entries of $P$ being zero. In particular, we calculate 
\begin{align*}
\IP(\tau_1>t_1,\tau_2>t_2) & = \begin{cases}
P_{(1,1),(1,1)}(t_1)\,\big( P_{(1,1),(1,1)}(t_2-t_1)+P_{(1,1),(0,1)}(t_2-t_1)\big) & , t_2 \geq t_1 \\
P_{(1,1),(1,1)}(t_2)\,\big( P_{(1,1),(1,1)}(t_1-t_2)+P_{(1,1),(1,0)}(t_1-t_2)\big) & , t_1 > t_2 
\end{cases}\\
& = \begin{cases}
\frac{\lambda_2-\tilde{\lambda}_2}{\lambda_1+\lambda_2-\tilde{\lambda}_2}\,e^{-(\lambda_1+\lambda_2)\,t_2}+\frac{\lambda_1}{\lambda_1+\lambda_2-\tilde{\lambda}_2} e^{-\tilde{\lambda}_2\,t_2-(\lambda_1+\lambda_2-\tilde{\lambda}_2)\,t_1} & , t_2 \geq t_1 \\
\frac{\lambda_1-\tilde{\lambda}_1}{\lambda_1+\lambda_2-\tilde{\lambda}_1}\,e^{-(\lambda_1+\lambda_2)\,t_1}+\frac{\lambda_2}{\lambda_1+\lambda_2-\tilde{\lambda}_1} e^{-\tilde{\lambda}_1\,t_1-(\lambda_1+\lambda_2-\tilde{\lambda}_1)\,t_2} & , t_1 > t_2 
\end{cases}.
\end{align*}   
The latter distribution is precisely the Freund distribution, which can be seen by comparing it to Equation (47.26) in \cite[p.\ 356]{kotz00}. We would like to note additionally that the so-called $ACBVE(\eta_1,\eta_2,\eta_{12})$-distribution of \cite{block74} arises as the three-parametric subfamily of the Freund distribution, obtained from the parameters 
\begin{gather*}
\lambda_1=\eta_1+\frac{\eta_{12}\eta_1}{\eta_1+\eta_2},\,\lambda_2=\eta_2+\frac{\eta_{12}\eta_2}{\eta_1+\eta_2},\,\tilde{\lambda}_1=\eta_1+\eta_{12},\,\tilde{\lambda}_2=\eta_2+\eta_{12}.
\end{gather*}
Multivariate extensions of the described Markov chain construction are now clearly straightforward. One simply has to define the intensity matrix $Q$ as follows: For each set $I \subset \{1,\ldots,d\}$ of non-bankrupt names one has to define exponential rates $\tilde{\lambda}_J$ for all subsets $J\subset I$ and $|J|=|I|-1$, and write them in the respective entry $Q_{h(I),h(J)}$. All other off-diagonal entries of $Q$ are set to zero, and then the diagonal elements are computed as the negative of the sum over all previously defined row entries. Similarly, one can generalize the model to allow for multiple defaults. This means that also subsets $J\subset I$ with $|J|=|I|-k$, $k\geq 1$, are assigned exponential rates.\par
For stepwise simulation along the $\Delta$-grid, required is nothing but the matrix $P[\Delta]=\exp(\Delta\,Q)$, which numerically is a standard routine (e.g.\ \texttt{expm} in MATLAB).
\section{A new characterization of the Marshall--Olkin law}\label{sec:newmarshallolkin}
Throughout this section, we denote by $\textbf{Z}_I$ the $|I|$-margin of the survival indicator process $\textbf{Z}$ which only consists of the components indexed by $I \subset \{1,\ldots,d\}$. 
Assuming the survival indicator process to be a Markovian process, there are still serious drawbacks to acknowledge with respect to practical viability:
\begin{itemize}
\item[(a)] If we want to carry out a simulation study involving only a subportfolio, i.e.\ a subgroup $I \subsetneq \{1,\ldots,d\}$ of components, we still have to simulate the whole portfolio $\textbf{Z}$ and cannot simulate the subvector $\textbf{Z}_I$ directly with a more efficient simulation engine. Hence, the model is not stable under taking margins, a property that is crucial for large portfolios that are frequently restructured.
\item[(b)] If our application requires us to add (remove) components to (from) our portfolio on a frequent basis, every such change might alter the dependence structure between the original components, and therefore requires careful readjustments of the model. In other terms, the model cannot be incremented straightforwardly in size in a nested fashion. Models with the property of being variable in the dimension are very manageable and popular. A typical case is the Gaussian one-factor copula model. 
\item[(c)] As a consequence of (b), in particular, the univariate marginal laws might be affected heavily by the dependence structure between all components when updating or re-balancing our portfolio. This means that it is unnatural to split the dependence structure from the margins. An application in which pre-assigned univariate models are coupled by an arbitrary dependence model a posteriori, a popular industry practice related to the use of copula functions, is not recommended in this respect.
\end{itemize}
In order to maintain all properties required for the aforementioned applications, one therefore has to postulate that all subprocesses $\textbf{Z}_I$ are Markovian, and not only $\textbf{Z}=\textbf{Z}_{\{1,\ldots,d\}}$. One observation is already helpful in this regard: If $\textbf{Z}$ is time-homogeneous Markovian, it is a continuous-time Markov chain on the finite state space $\{0,1\}^d$. Since the distributional properties of these processes are well-known to be characterized in terms of a finite-dimensional intensity matrix, the distribution of $(\tau_1,\ldots,\tau_d)$ must also be from a finite-parametric family. The following distribution clearly is a candidate, named after the seminal reference \cite{marshall67}.
\begin{definition}[The Marshall--Olkin distribution]
On a probability space $(\Omega,\F,\IP)$, a random vector $(\tau_1,\ldots,\tau_d)$ taking values in $[0,\infty)^d$ is said to follow a Marshall--Olkin distribution if it has the survival function
\begin{gather*}
\IP(\tau_1>t_1,\ldots,\tau_d>t_d)=\exp\Big(-{\sum_{\emptyset \neq I \subset \{1,\ldots,d\}}\lambda_I \,\max_{i \in I}\{t_i\}}\Big),\quad t_1,\ldots,t_d \geq 0,
\end{gather*}
for non-negative parameters $\{\lambda_I\}$, $\emptyset \neq I \subset \{1,\ldots,d\}$, satisfying $\sum_{I:i \in I}\lambda_I>0$ for all $i=1,\ldots,d$. A canonical construction for this distribution is based on a collection of independent exponential random variables $\{E_I\}$ with $E_I \sim Exp(\lambda_I)$ and given via
\begin{gather}\label{MO_construct}
\tau_k:=\min\{E_I\,:\,k \in I\},\quad k=1,\ldots,d.
\end{gather}
\end{definition}

We now focus on the multi-variate case and formally prove an intuitive statement which intimately links the lack-of-memory property of a random vector to the Markovianity of its associated survival indicator process. Connections between Markov chains and random vectors have already been studied in the statistical literature, both implicitly and explicitly. For instance, there is a branch of literature in reliability theory concerned with multi-variate exponential distributions, motivated by multi-variate versions of the lack-of-memory property. Among these the Marshall--Olkin distribution is the most popular one, because from many viewpoints - and we present another one in Theorem \ref{main_thm} - it satisfies an intuitive and useful lack-of-memory property. Most dominantly, it is stable under marginalization, i.e.\ lower-dimensional margins satisfy the same lack-of-memory property as well and, in particular, the univariate margins are exponential. When giving up this stability property, but still postulating a multi-variate lack-of-memory property,  one can still have random vectors whose associated survival indicators are Markovian (even though to the best of our knowledge, this has never been observed explicitly in the literature before)\footnote{Even more general is the family of multi-variate phase type distributions introduced in \cite{assaf84}, see also \cite{cai05}, which define a random vector explicitly via a Markov chain. The default times are defined as the first time points at which an underlying Markov chain reaches an absorbing state, and thus also serve as a very intuitive framework for credit risk modeling - thinking about the link with credit rating transition matrices. Unfortunately, multi-variate phase type distributions, due to their generality, appear to be very difficult to work with.}. Summing up, we have the following proper inclusion: 
\begin{align*}
\{(\tau_1,\ldots,\tau_d )\sim \text{ Marshall--Olkin law}\}&\subset \{\textbf{Z}(t) \text{ is time-homogeneous Markovian}\} \\
\text{which we describe as} \\
 \text{``stepwise simulation'' }\oplus\text{ ``marginalization''} &\subset \text{``stepwise simulation''}
\end{align*}
With regards to the first inclusion, Theorem \ref{main_thm} shows that the Marshall--Olkin distribution arises as the proper subset of random vectors whose Markov property is preserved under marginalization.

\newpage

\begin{theorem}[Markovianity of survival indicators and lack-of-memory]\label{main_thm}
The $d$-dimensional survival indicator processes $\textbf{Z}_I$ are time-homogeneous Markovian for all subsets $\emptyset\neq I \subset \{1,\ldots,d\}$ $\Leftrightarrow$
 $(\tau_1,\ldots,\tau_d)$ has a Marshall--Olkin distribution
\end{theorem}
\begin{proof}
``$\Rightarrow$'' By the time-homogeneous Markov property, {there is a transition function $p_{\textbf{x},\textbf{y}}(t)$ for $\textbf{x},\textbf{y} \in \{0,1\}^d$ and $t \geq 0$ such that
\begin{gather*}
\IP(\textbf{Z}(t_n)=\textbf{x}_n,\ldots,\textbf{Z}(t_1)=\textbf{x}_1) = p_{(1,\ldots,1),\textbf{x}_1}(t_1)\,\prod_{l=2}^{n}p_{\textbf{x}_{l-1},\textbf{x}_l}(t_l-t_{l-1})
\end{gather*}
for $t_n>\ldots>t_1>0$ and $\textbf{x}_1,\ldots,\textbf{x}_n \in \{0,1\}^d$. Let $t,s_1,\ldots,s_d \geq 0$ be arbitrary and denote by $\pi$ a permutation such that $s_{\pi(1)}\leq s_{\pi(2)} \leq \ldots \leq s_{\pi(d)}$ is the ordered list of $s_1,\ldots,s_d$. Define the following subsets of $\{0,1\}^d$:
\begin{gather*}
A_1:=\{(1,\ldots,1)\},\quad A_k:=\big\{ \textbf{x} \in \{0,1\}^d\,:\,x_{\pi(l)}=1 \mbox{ for all }l \geq k\big\},\quad k=2,\ldots,d.
\end{gather*}
In words, $A_k$ denotes the subset of $\{0,1\}^d$ in which all components $\pi(k),\ldots,\pi(d)$ are still alive. There is a finite number $N$ of distinct paths $(\textbf{x}^{(i)}_2,\ldots,\textbf{x}^{(i)}_d) \in A_2 \times \ldots \times A_d$, $i=1,\ldots,N$, that avoid inconsistent patterns in time (such as default resurrections etc), i.e. such that
\begin{align*}
0 < \IP(\textbf{Z}(t+s_{\pi(1)})=(1,\ldots,1),\textbf{Z}(t+s_{\pi(2)})=\textbf{x}^{(i)}_2,\ldots,\textbf{Z}(t+s_{\pi(d)})=\textbf{x}^{(i)}_d).
\end{align*}
This set of paths depends on $s_1,\ldots,s_d$, but it does not depend on $t$ by the time-homogeneity property of $\textbf{Z}$. 
We have
\begin{align*}
& \IP(\tau_1>t,\ldots,\tau_d>t)\,\IP(\tau_1>s_1,\ldots,\tau_d>s_d)\\
& \quad =\IP(\textbf{Z}(t) \in A_1)\,\IP\big( \textbf{Z}(s_{\pi(1)}) \in A_1,\,\textbf{Z}(s_{\pi(2)}) \in A_2,\ldots,\textbf{Z}(s_{\pi(d)}) \in A_d\big) \\
& \quad =\IP(\textbf{Z}(t) \in A_1)\,\sum_{i=1}^{N}\IP(\textbf{Z}(s_{\pi(1)})=(1,\ldots,1),\textbf{Z}(s_{\pi(2)})=\textbf{x}^{(i)}_2,\ldots,\textbf{Z}(s_{\pi(d)})=\textbf{x}^{(i)}_d)\\
& \quad =p_{(1,\ldots,1),(1,\ldots,1)}(t)\,\sum_{i=1}^{N}p_{(1,\ldots,1),(1,\ldots,1)}(s_{\pi(1)})\,p_{(1,\ldots,1),\textbf{x}^{(i)}_2}(s_{\pi(2)}-s_{\pi(1)})\,\prod_{k=3}^{d}p_{\textbf{x}^{(i)}_{k-1},\textbf{x}^{(i)}_k}(s_{\pi(k)}-s_{\pi(k-1)})\\
& \quad =\sum_{i=1}^{N} p_{(1,\ldots,1),(1,\ldots,1)}(t+s_{\pi(1)})\,p_{(1,\ldots,1),\textbf{x}^{(i)}_2}(t+s_{\pi(2)}-(t+s_{\pi(1)}))\,\prod_{k=3}^{d}p_{\textbf{x}^{(i)}_{k-1},\textbf{x}^{(i)}_k}(t+s_{\pi(k)}-(t+s_{\pi(k-1)}))\\
& \quad = \IP(\textbf{Z}(t+s_{\pi(1)})\in A_1,\textbf{Z}(t+s_{\pi(2)}) \in A_2,\ldots,\textbf{Z}(t+s_{\pi(d)}) \in A_d)\\
& \quad = \IP(\tau_1>t+s_1,\ldots,\tau_d>t+s_d)
\end{align*}
}
Repeating the above derivation for every subset $I \subset \{1,\ldots,d\}$ we obtain the equation
\begin{gather*}
\IP(\tau_{i_1}>t+s_{i_1},\ldots,\tau_{i_k}>t+s_{i_k})=\IP(\tau_{i_1}>t,\ldots,\tau_{i_k}>t)\,\IP(\tau_{i_1}>s_{i_1},\ldots,\tau_{i_k}>s_{i_k})
\end{gather*}
for arbitrary $1 \leq i_1,\ldots,i_k \leq d$ and $t,s_{i_1},\ldots,s_{i_k} \geq 0$. This is precisely the functional equality describing the multi-variate lack-of-memory property, which is well-known to characterize the Marshall--Olkin exponential distribution, see \cite{marshall67,marshall95}.
\par

``$\Leftarrow$'' Assume $(\tau_1,\ldots,\tau_d)$ has a Marshall--Olkin distribution with parameters $\{\lambda_I\}$, $\emptyset \neq I \subset \{1,\ldots,d\}$ satisfying $\sum_{I:k \in I}\lambda_I>0$ for all $k=1,\ldots,d$. We prove Markovianity of $\textbf{Z}_I$ for an arbitrary non-empty subset $I$ of components. Without loss of generality, we may assume that $(\tau_1,\ldots,\tau_d)$ is defined on the following probability space, as first considered in \cite{arnold75}: we consider an iid sequence $\{{E}_n\}_{n \in \IN}$ of exponential random variables with rate $\lambda := \sum_{\emptyset\neq K \subset \{1,\ldots,d\}}\lambda_K$ and an independent iid sequence $\{Y_n\}_{n \in \IN}$ of set-valued random variables with distribution given by
\begin{gather*}
\IP(Y_1=K)=p_K:=\frac{\lambda_K}{\lambda}, \quad \emptyset \neq K \subset \{1,\ldots,d\}.
\end{gather*}
The random vector $(\tau_1,\ldots,\tau_d)$ is then defined as $\tau_k:={E}_1+\ldots+{E}_{\min\{n\,:\,k \in Y_n\}}$, $k=1,\ldots,d$. 
Let us introduce the notation
\begin{gather*}
N_t:=\sum_{k=1}^{\infty}1_{\{{E}_1+\ldots+{E}_k \leq t\}},\quad t \geq 0,
\end{gather*}
which is a Poisson process with intensity $\lambda$.
Fix a non-empty set $I \subset \{1,\ldots,d\}$, say $I=\{i_1,\ldots,i_k\}$ with $1 \leq i_1<\ldots<i_k\leq d$. Denoting the power set of $\{1,\ldots,d\}$ by $\mathcal{P}_d$, we define the function $f_I:\{0,1\}^k \times \mathcal{P}_d \rightarrow \{0,1\}^k$ as follows:
\begin{gather*}
j\mbox{-th component of }f_I(\vec{x},J):=1_{\{x_j=1 \mbox{ and }i_j \notin J\}},\quad j=1,\ldots,k,
\end{gather*}
for $\vec{x}=(x_1,\ldots,x_k) \in \{0,1\}^k$ and $J \in \mathcal{P}_d$. It is now readily observed -- in fact just a rewriting of Arnold's model -- that 
\begin{gather}
\textbf{Z}_I(t) = f_I\Big(\textbf{Z}_I(s),\bigcup_{k=N_s+1}^{N_t}Y_k \Big),\quad t \geq s \geq 0.
\label{arnoldrepresent}
\end{gather}
This stochastic representation implies the claim, since the second argument of $f_I$ is independent of $\F_I(s):=\sigma(\textbf{Z}_I(u)\,:\,u \leq s)$ by the Poisson property of $\{N_t\}$. To see this, it suffices to observe that $Z_I(s)$ is a function of $N_s$ and $Y_{1},\ldots,Y_{N_s}$ (which can be seen by setting $t=s$ and $s = 0$ in (\ref{arnoldrepresent})), whereas the second argument is a function of $Y_{N_s+1},\ldots,Y_{N_t}$. Consequently, the independent random variables $N_s$ and $N_t-N_s$ only serve as a random pick of two independent (because disjoint) partial sequences of the iid sequence $Y_1,Y_2,\ldots$. 
\end{proof}

\newpage

\subsection{Parameterization and efficient implementation}
The efficient implementation of an unbiased simulation scheme for the Marshall--Olkin law is subject of this paragraph. We consider the tasks:
\begin{itemize}
\item[(a)] Finding a convenient parameterization of the Marshall--Olkin law, especially in large dimensions.
\item[(b)] Constructing an efficient and unbiased simulation engine for the Marshall--Olkin law along a given time grid $0<t_1<t_2<\ldots<t_n=T$.
\end{itemize}
Simulation requires a (preferably simple) stochastic model. There exist two classical stochastic representations for the Marshall--Olkin distribution. The first, see (\ref{MO_construct}), requires $2^d-1$ exponentially distributed shocks, see \cite{marshall67}. A second, see \cite{arnold75}, is based on compound sums of exponentials. In both models, the tasks (a) and (b) are intimately linked, because the number of parameters different from zero enters the (expected) runtime of the respective simulation algorithms, see \cite[Chapter 3.1]{maibook12}. 

The references \cite{giesecke03,lindskog03,burtschell09} tackle this issue by setting most parameters to zero, however, this results in simplistic subfamilies. Concluding, these canonical stochastic models are not recommended in dimensions greater than, say, $d=10$, although occasionally the dynamical properties of the aggregated default counting process and of the related loss distribution have been studied under pool homogeneity assumptions, see for example \cite{brigo07}, in dimensions such as 125. We should also mention that, in this context,  \cite{crepeycdo} manages to attain high dimensions and realistic calibration to market data while modeling defaults in a bottom-up fashion and with no need to assume pool homogeneity. 
\par 
There exists, however, a third stochastic representation of the Marshall--Olkin distributions due to \cite{maiJMVA09,maiJSCS10}, based on the notions of L\'evy subordinators. It is thus called ``L\'evy-frailty construction''. This approach has been generalized and applied to portfolio-credit risk by, e.g., \cite{bernhart11,linetsky12}. Although the (one-factor) L\'evy-frailty construction does not include all possible Marshall--Olkin laws, it is still general enough to comprise a subfamily that is big enough for applications. Furthermore, with regards to the tasks (a) and (b), it has three crucial advantages:
\begin{itemize}
\item[(1)] The number of parameters does not depend on the dimension, but instead can be chosen quite arbitrarily.
\item[(2)] The stochastic model can be interpreted as a factor model, such that a simulation along a given time grid is natural and straightforward. The numerical effort increases only linearly in the dimension $d$ and the number of time steps of the grid.
\item[(3)] The model serves as a convenient building block for hierarchical (and other) factor models. This will be subject of Proposition \ref{survival_factor}.
\end{itemize}
The L\'evy-frailty construction for the random times $(\tau_1,\ldots,\tau_d)$ works as follows. Define
\begin{gather*}
\tau_k:=\inf\big\{t\geq 0: \Lambda_t\geq E_k\big\},\quad k=1,\ldots,d,
\end{gather*}
where the sequence $\{E_k\}_{k\in\IN}$ consists of iid unit exponentials and the independent L\'evy subordinator $\Lambda=\{\Lambda_t\}_{t\geq 0}$ is characterized by its Laplace exponent $\Psi:\IR_+\rightarrow \IR_+$ via $\IE[\exp(-x\,\Lambda_t)]=\exp(-t\,\Psi(x))$, for all $x,t\geq 0$. The L\'evy subordinator acts as a joint factor on the independent list of exponential trigger variables $E_1,\ldots,E_d$. The resulting $\tau_k$'s are defined as the first passage times of the L\'evy subordinator across the individual trigger variables. Jumps in the L\'evy subordinator imply the possibility of multiple components being killed at the same time. The lack-of-memory property of the Marshall--Olkin distribution is the result of the lack-of-memory property of the exponential trigger variates and the independent and stationary increments of the L\'evy subordinator. This property will be exploited for simulations later on. It can be shown, see \cite{maiJMVA09}, that the choice of L\'evy subordinator affects the homogeneous $Exp(\lambda)$-marginal laws via $\lambda=\Psi(1)$ and implies a Marshall--Olkin survival function of the form
\begin{gather*}
\IP(\tau_1>t_1,\ldots,\tau_d>t_d)=\prod_{k=1}^d e^{-(t_{\pi(k)}-t_{\pi(k-1)})\,\Psi(d+1-k)},
\end{gather*}
where $t_{\pi(1)} \leq t_{\pi(2)} \leq \ldots \leq t_{\pi(d)}$ is the ordered list of $t_1,\ldots,t_d \in [0,\infty)$. The structural properties of the resulting dependence structure relate to properties of the L\'evy subordinator $\Lambda$. For the construction of parametric models it is an important observation that each parametric family of L\'evy subordinators, given in terms of the Laplace exponent $\Psi_{\theta}$, implies a parametric subfamily of the Marshall--Olkin law with the same parameter vector $\theta$. This provides a convenient, yet flexible, methodology to set up models with a reasonable number of parameters. Moreover, sampling is intuitive and fast, as demonstrated in Algorithm \ref{algo1}.
\begin{algorithm}[Sampling the L\'evy-frailty model on a given time grid]\label{algo1}
Given the time grid $0=t_0<t_1<\ldots<t_n=T$ and the L\'evy subordinator $\Lambda$. Initialize the current time as $t^{\ast}:=0$, $\ell:=0$, and the number of components that are still alive by $nalive=d$.
\begin{itemize}
\item[(1)] Repeat the following steps until $\big((nalive==0)$ or $(t^{\ast}==T)\big)$, i.e.\ until all components are destroyed or the final time horizon is reached, whichever takes place first:
\begin{itemize}
\item[(a)] Set $t^{\ast}:=t_{\ell+1}$.
\item[(b)] Simulate the next increment $\Delta \Lambda_{\ell}:=\Lambda_{t_{\ell+1}}-\Lambda_{t_{\ell}}\sim \Lambda_{t_{\ell+1}-t_{\ell}}$ of the L\'evy subordinator on the time interval $[t_{\ell},t_{\ell+1}]$. Note that this is independent of the past, by the L\'evy properties of $\Lambda$.
\item[(c)] Simulate a list of independent unit exponentials $E_{i_1},\ldots,E_{i_{nalive}}$ for the components $\tau_{i_1},\ldots,\tau_{i_{nalive}}$ that have not been killed, yet. This is justified by the lack-of-memory property of the unit exponential law, i.e.\ the positive distance of a trigger variate (that has not been killed, yet) to the current state of the subordinator has a unit exponential law.
\item[(d)] For each $E_{i_k}$, $k=1,\ldots, nalive$, test if $(\Delta \Lambda_{\ell}>E_{i_k})$. Each time this condition is met, set $\tau_{i_k}:=t^{\ast}$ and decrease $nalive$ by one.\footnote{Instead of drawing exponential random variables along the lines of the L\'evy-frailty model, one might instead use Bernoulli$(1-\exp(-\Delta \Lambda_{\ell}))$ distributed ones in Step (1)(c,d) of Algorithm \ref{algo1}. This is justified by the observation that the conditional default probability of $\tau_k$ within $[t_{\ell},t_{\ell+1}]$ given $\Delta \Lambda_{\ell}$ is precisely $1-\exp(-\Delta \Lambda_{\ell})$. If the respective Bernoulli experiment was succesful, component $t_k$ is killed and set to $\tau_k:=t^{\ast}$. }
\item[(e)] $\ell:=\ell+1$.
\end{itemize}
\item[(2)] Return the vector $(\tau_1,\ldots,\tau_d)\in\{t_0,\ldots,t_n\}^d$ or, equivalently, the path of the indicator process $(\I{\{\tau_1>t\}},\ldots,\I{\{\tau_d>t\}})$ sampled on the given time grid.
\end{itemize}
\end{algorithm}
Next, we discuss some parametric families. L\'evy subordinators can be specified via Bernstein functions $\Psi$, acting as Laplace exponent, or via the law of $\Lambda_1$; the latter must be non-negative and infinitely divisible. For both, several parametric families can be found in, e.g., \cite{bertoin99,schilling10}. For the applications we are aiming at, we need families with Laplace exponent $\Psi$ given in closed form (to evaluate the joint survival function) and an efficient simulation scheme for the increments $\Delta \Lambda_{\ell}$ (for Algorithm \ref{algo1}). Below, we briefly provide and discuss some examples.
\begin{itemize}
\item[(1)] The simplest subordinator is a linear drift $\Lambda_t=\mu\, t$, $\mu>0$. This implies independent $Exp(\mu)$-distributed $\tau_1,\ldots,\tau_d$. We can add a jump to infinity at a random time $E\sim Exp(\lambda)$, providing $\Lambda_t=\mu\,t+\infty\,\I{\{t>E\}}$. The interpretation is an Armageddon scenario at time $E$ that kills all remaining components, a model implicitly used in \cite{burtschell09}. The Laplace exponent is $\Psi(x)=\mu\,x+\lambda\,\I{\{x>0\}}$.
\item[(2)] Another example is a compound Poisson process with drift $\mu\geq 0$. In this case, $\Lambda_t=\mu\,t+\sum_{i=1}^{N_t}Y_i$, where $\{N_t\}_{t\geq 0}$ is a Poisson process with intensity $\lambda>0$ and the $Y_i$ are iid random variables on $[0,\infty)$. The number of jumps within time $\Delta$ follows a $Poi(\lambda\, \Delta)$ distribution, such that the simulation is straightforward whenever the jump-size distribution of the $Y_i$ can be simulated. The Laplace exponent is $\Psi(x)=\mu\,x+\lambda\,\IE[1-\exp(-x\,J_1)]$.
\item[(3)] The Gamma subordinator, parameterized by $\beta>0,\eta>0$, is another example. Its Laplace exponent is given by $\Psi(x)=\beta\,\log(1+x/\eta)$, its distribution at time $t$ is a Gamma distribution and can thus easily be simulated, see, e.g., \cite[Algorithms 6.5 and 6.6, p.\ 242--243]{maibook12}.
\item[(4)] The Inverse Gaussian subordinator is parameterized by $\beta>0,\eta>0$, with $\Psi(x)=\beta\,(\sqrt{2\,x+\eta^2}-\eta)$. Its distribution at time $t$ is the same as the one of the first hitting-time of the level $\beta\,t$ of a Brownian motion with drift, hence the name. A simulation strategy is provided in \cite[Algorithm 6.10, p.\ 245]{maibook12}.
\item[(5)] The stable subordinator with parameter $\alpha\in(0,1]$ has a Laplace exponent given by $\Psi(x)=x^\alpha$. Its increments can be sampled as shown, e.g., in \cite[Algorithm 6.11, p.\ 246]{maibook12}.
\end{itemize}
Independent L\'evy subordinators form a cone with Laplace exponent being the corresponding linear combination of the Laplace exponents of the building blocks. Similarly, subordinators are stable under independent subordination; providing a second way to construct new subordinators from given ones. In particular, this allows to increase the number of parameters. In both cases, the simulation is possible whenever the building blocks can be simulated.
\par
The (one-factor) L\'evy-frailty construction provides subfamilies of the Marshall--Olkin law that are conditionally iid. This can be generalized without giving up numerical tractability. The simplest generalization is to alter the unit exponential trigger variables $E_k$ to exponentials with individual rate $\lambda_k$. The implementation of this conditional independence approach similar to Algorithm \ref{algo1} is immediate. Generalizing the dependence structure to non-homogeneous structures is possible via a factor-model ansatz, see \cite{linetsky12}. Starting from $m$ independent L\'evy subordinators $\hat{\Lambda}^{(1)},\ldots,\hat{\Lambda}^{(m)}$ with Laplace exponents $\hat{\Psi}_1,\ldots,\hat{\Psi}_m$ and considering the weight vectors $\bm{\theta}_k\in\IR^m_+$, $k=1,\ldots,d$, we can define the $d$-dimensional subordinator $\bm{\Lambda}=(\Lambda^{(1)},\ldots,\Lambda^{(d)})$, where $\Lambda^{(k)}=\bm{\theta}_k'\,\hat{\bm{\Lambda}}$. To simulate from this model, it is sufficent to simulate in Step (1)(b) of Algorithm \ref{algo1} each increment of the $m$ individual subordinators and to compute the resulting required subordinator increments thereof. This requires about $m$ times the effort of Algorithm \ref{algo1}, which still increases only linearly in the dimension $d$. 
\begin{proposition}[Joint survival function in the multifactor model]\label{survival_factor}
Let $t_1\geq 0,\ldots,t_d\geq 0$ and denote by $t_{\pi(1)}\leq \ldots\leq t_{\pi(d)}$ the ordered list, $\pi:\{1,\ldots,d\}\rightarrow \{1,\ldots,d\}$ is the permutation map that corresponds to this ordering. The joint survival function of $(\tau_1,\ldots,\tau_d)$ is then given by
\begin{gather*}
\IP(\tau_1>t_1,\ldots,\tau_d>t_d)=\exp\Big(-\sum_{\ell=1}^m\sum_{j=1}^d\hat{\Psi}_{\ell}\big(\sum_{k=j}^d \theta_{\ell,\pi(k)}\big)(t_{\pi(j)}-t_{\pi(j-1)})\Big).
\end{gather*}
\end{proposition}
\begin{proof}
We condition on all $m$ factors, providing conditional independence of the default times, and rewrite each of the L\'evy subordinators that have to be integrated out as a sum of its independent increments. This requires a combinatorial rearrangement to see $(\ast)$. Finally, writing the appearing Laplace transforms in terms of the respective Laplace exponents provides the result.
\begin{align*}
&\IP(\tau_1>t_1,\ldots,\tau_d>t_d)\\
&\quad =\IE\Big[\IP\big(\tau_1>t_1,\ldots,\tau_d>t_d\big|\sigma(\{\hat{\Lambda}^{(\ell)}_t\}: \ell=1,\ldots,m, t\in[0,t_{\pi(d)}])\big)\Big]\\
&\quad =\IE\Big[e^{-\sum_{k=1}^d\Lambda^{(k)}_{t_k}}\Big] =\IE\Big[e^{-\sum_{k=1}^d\sum_{\ell=1}^m\theta_{\ell,k}\hat{\Lambda}^{(\ell)}_{t_k}}\Big] =\IE\Big[e^{-\sum_{\ell=1}^m\sum_{k=1}^d\theta_{\ell,k}\hat{\Lambda}^{(\ell)}_{t_k}}\Big]\\
&\quad =\prod_{\ell=1}^m\IE\Big[e^{-\sum_{k=1}^d\theta_{\ell,k}\hat{\Lambda}^{(\ell)}_{t_k}}\Big] =\prod_{\ell=1}^m\IE\Big[e^{-\sum_{k=1}^d\theta_{\ell,k}\sum_{j=1}^{\pi^{-1}(k)}(\hat{\Lambda}^{(\ell)}_{t_{\pi(j)}}-\hat{\Lambda}^{(\ell)}_{t_{\pi(j-1)}})}\Big]\\
&\quad \stackrel{(\ast)}{=}\prod_{\ell=1}^m\IE\Big[e^{-\sum_{j=1}^d\sum_{k=j}^{d}\theta_{\ell,\pi(k)}(\hat{\Lambda}^{(\ell)}_{t_{\pi(j)}}-\hat{\Lambda}^{(\ell)}_{t_{\pi(j-1)}})}\Big]\\
&\quad =\exp\Big(-\sum_{\ell=1}^m\sum_{j=1}^d\hat{\Psi}_{\ell}\big(\sum_{k=j}^d \theta_{\ell,\pi(k)}\big)(t_{\pi(j)}-t_{\pi(j-1)})\Big).
\end{align*}
\end{proof}
Closely related, a hierarchical and extendible Marshall--Olkin law is constructed in \cite{maiJSCS10,MaiScherer11}. The idea behind is to group the components according to some economic criterion (geographic region, industry segment, etc.). All components are affected by some global factor. Additionaly, group specific factors add further dependence to all components within some group. The result is a hierarchical structure in which the dependence within each group is larger than the dependence between the groups. With regard to our factor structure, this is achieved for $J$ groups and $m=J+1$ subordinators via the weights $\bm{\theta}_{k,j}=(\alpha_j,0,\ldots, 0,\beta_j,0,\ldots,0)\in\IR_+^{J+1}$. 
\begin{remark}[Constructing the full Marshall--Olkin class] The multi-factor L\'evy-frailty construction is general enough to comprise the full family of Marshall--Olkin distributions. To this end, we use $m=2^d-1$ independent killed subordinators $\hat{\Lambda}^{(I)}_t:=\infty\,\I{\{t>E_I\}}$ and $\Lambda^{(k)}_t:=\sum_{I:k\in I}\hat{\Lambda}^{(I)}_t$, which is basically just a complicated way of writing the original Marshall--Olkin shock model (\ref{MO_construct}). This construction is not unique and provides an alternative proof of \cite[Theorem 4.2]{linetsky12}.
\end{remark}
\section{Case study: Simulation bias for selected multi-variate distributions}\label{sec:casestudy}
The case study in this section illustrates how wrong things can go when carelessly assuming the equivalence of the simulation approaches (i) and (ii). This had been pointed out with a numerical example already in \cite{BrigoChourdakisSelfChaining}, but in a more limited context where time steps were all equal and where only the full joint survival was considered. 
To illustrate this effect in our more general framework, considering two default times $(\tau_1,\tau_2)$ is sufficient. An approach commonly used in practice is to model the marginal survival functions $\bar{F}_1$ and $\bar{F}_2$ of $\tau_1$ and $\tau_2$ separately, and link them by a certain survival copula $C$ afterwards, the most prominent example being the Gaussian copula. The marginal laws $F_1$ and $F_2$ are assumed to be exponential in our case study, because the lack-of-memory property of the exponential distribution is a necessary requirement for the validity of the following stepwise simulation algorithm already for the univariate marginals. In order to simulate this bivariate model stepwise, we run the following algorithm:
\begin{algorithm}[Stepwise simulation of bivariate survival indicator process] \label{algoSimu}
\begin{enumerate}
\item Simulate a vector $(X_1,X_2)\sim C(\bar{F}_1,\bar{F}_2)$ and compute the indicator $(I_1,I_2):=(\I{\{X_1>\Delta\}},\I{\{X_2>\Delta\}})$. Set $\textbf{Z}(\Delta):=(I_1,I_2)$.
\item Simulate a vector $(X_1,X_2)\sim C(\bar{F}_1,\bar{F}_2)$ and compute the indicator $(I_1,I_2):=(\I{\{X_1>\Delta\}},\I{\{X_2>\Delta\}})$. Set $\textbf{Z}(2\,\Delta):=(\I{\{Z_1(\Delta)=1,\,I_1=1\}},\I{\{Z_2(\Delta)=1,\,I_2=1\}})$.
\item Simulate a vector $(X_1,X_2)\sim C(\bar{F}_1,\bar{F}_2)$ and compute the indicator $(I_1,I_2):=(\I{\{X_1>\Delta\}},\I{\{X_2>\Delta\}})$. Set $\textbf{Z}(3\,\Delta):=(\I{\{Z_1(2\,\Delta)=1,\,I_1=1\}},\I{\{Z_2(2\,\Delta)=1,\,I_2=1\}})$.
\item ...
\end{enumerate}
\end{algorithm}
The output of this algorithm is interpreted as a (discretized) path $(\textbf{Z}(\Delta),\textbf{Z}(2\,\Delta),\textbf{Z}(3\,\Delta),\ldots)$ of the survival indicator process $\textbf{Z}(t)=(\I{\{\tau_1>t\}},\I{\{\tau_2>t\}})$. However, this is not always appropriate, which is what the present exercise illustrates. In particular, according to Theorem \ref{main_thm} it is appropriate only if the joint distribution of $(\tau_1,\tau_2)$ has a Marshall--Olkin distribution. Let us make two remarks about common sources of errors:
\begin{itemize}
\item Plugging exponential marginals $F_1,F_2$ into an arbitrary Marshall--Olkin copula does not necessarily yield a bivariate Marshall--Olkin distribution. This is a massive difference from the Gaussian world, indicating that separation of marginals and dependence structure is not always straightforward. Indeed, the bivariate Marshall--Olkin distribution has three non-negative parameters $\lambda_{\{1\}},\lambda_{\{2\}},\lambda_{\{1,2\}}$ satisfying $\lambda_{\{i\}}+\lambda_{\{1,2\}}>0,\,i=1,2$. It is divided into two exponential marginals $F_i=Exp(\lambda_{\{i\}}+\lambda_{\{1,2\}})$, $i=1,2$, and survival copula of the form
\begin{gather*}
C(u,v) = \min\Big\{ v\,u^{1-\frac{\lambda_{\{1,2\}}}{\lambda_{\{1\}}+\lambda_{\{1,2\}}}},u\,v^{1-\frac{\lambda_{\{1,2\}}}{\lambda_{\{2\}}+\lambda_{\{1,2\}}}}\Big\},
\end{gather*}
in the sense that $\IP(\tau_1>t_1,\tau_2>t_2)=C\big(\IP(\tau_1>t_1),\IP(\tau_2>t_2)\big)$, with $\IP(\tau_i>t_i)=\exp\big(-(\lambda_{\{i\}}+\lambda_{\{1,2\}})t_i\big)$. In principle, the copula is only two-parametric, namely determined by the two auxiliary parameters
\begin{gather*}
\alpha:=\frac{\lambda_{\{1,2\}}}{\lambda_{\{1\}}+\lambda_{\{1,2\}}} \in [0,1],\quad \beta := \frac{\lambda_{\{1,2\}}}{\lambda_{\{2\}}+\lambda_{\{1,2\}}} \in  [0,1],
\end{gather*} 
as proposed for instance in the textbooks \cite{nelsen99,embrechts05}. Indeed, for each given pair $(\alpha,\beta) \in [0,1]^2$ we can find parameters $(\lambda_{\{1\}},\lambda_{\{2\}},\lambda_{\{1,2\}})$ of a Marshall--Olkin distribution yielding the desired pair $(\alpha,\beta)$. But then, in order for the joint law of $(\tau_1,\tau_2)$ to be of a proper Marshall--Olkin kind, the exponential rates of the marginals are restricted to the values $\lambda_{\{i\}}+\lambda_{\{1,2\}}$, for admissible Marshall--Olkin parameters $\lambda_{\{1\}},\lambda_{\{2\}},\lambda_{\{1,2\}}$ matching the given $(\alpha,\beta)$. If not, we obtain a multi-variate distribution violating the lack-of-memory property, and therefore the multi-variate indicator process looses the Markov property.
\item The article \cite{BrigoChourdakisSelfChaining} finds that for Algorithm \ref{algoSimu} to yield an unbiased sample of the survival indicator path, the copula $C$ needs to be a so-called extreme-value (also called self-chaining)  copula, i.e.\ it satisfies $C(u^t,v^t)=C(u,v)^t$ for each $t \geq 0$. This family of copulas is a proper superclass of the family of Marshall--Olkin copulas, and as a prominent example it also includes the Gumbel copula. However, our simulation study shows that this holds only for equal time steps and only until a first default happens. After that, this result is no longer valid and the Gumbel copula leads to simulation biases as well. The difference of our more general approach with respect to \cite{BrigoChourdakisSelfChaining} is that here we consider  the bivariate survival indicator process $\textbf{Z}(t)=(\I{\{\tau_1>t\}},\I{\{\tau_2>t\}})$ rather than  the one-dimensional survival indicator $\I{\{\tau_1>t,\tau_2>t\}}$. The latter is Markovian, and hence stepwise simulation feasible, if and only if $\min\{\tau_1,\tau_2\}$ is exponential. This property is indeed satisfied by all extreme-value/self-chaining copulas and even by the larger class of copulas obtained from multi-variate laws with exponential minima, see \cite{esary74}.   
\end{itemize}
\subsection{The case study}
We consider a random vector $(\tau_1,\tau_2)$ with bivariate survival copula $C$ and exponential margins with parameters $\lambda_1,\lambda_2$. We compute the probability $\IP(\tau_1>T,\tau_2>S)$ for $T=S=10$ and for $S=T/2=5$ in three different ways:
\begin{itemize}
\item[(a)] Exactly, using the formula $C\big(\exp({-\lambda_1 T}),\exp({-\lambda_2 S})\big)$.
\item[(b)] Via $n$ iid simulations of $(\I{\{\tau_1>T\}},\I{\{\tau_2>S\}})$ and the empirical frequency. This direct simulation approach is only included in order to test the validity of our simulation engine.
\item[(c)] Via $n$ iid simulations of $(\I{\{\tau_1>T\}},\I{\{\tau_2>S\}})$ and the empirical frequency, where the simulation of $(\I{\{\tau_1>T\}},\I{\{\tau_2>S\}})$ is carried out stepwise in two steps, i.e.\ by Algorithm \ref{algoSimu} with $\Delta=5$. We seek to illustrate that this is only justified for extreme-value copulas in the case $T=S$ and only for Marshall--Olkin distributions for $T \neq S$.
\end{itemize}
The above computations are carried out for three different survival copulas $C$:
\begin{itemize}
\item[(1)] A Marshall--Olkin copula $C_{\alpha}(u_1,u_2)=\min(u_1,u_2)\,\max(u_1,u_2)^{1-\alpha}$, which is such that the resulting joint distribution function is a proper bivariate Marshall--Olkin distribution.
\item[(2)] A Gumbel copula $C_{\theta}(u_1,u_2)=\exp(-(\log(1/u_1)^{1/\theta}+\log(1/u_2)^{1/\theta})^\theta)$, which is an extreme-value copula and, at the same time,  an Archimedean copula. Here, we expect method (c) to fail in the case $S<T$ (because the bivariate survival indicator process is not Markovian) but to work in the case $S=T$ (because the one-dimensional survival indicator process $\I{\{\min\{\tau_1,\tau_2\}>t\}}$ is Markovian by the extreme-value copula property).
\item[(3)] A Gaussian copula $C_{\rho}(u_1,u_2)=\mathcal{N}_{2}(\Phi^{-1}(u_1),\Phi^{-1}(u_2);\rho)$, where $\mathcal{N}_{2}(.,.;\rho)$ is the cdf of the bivariate normal distribution with mean vecor zero and correlation $\rho\in(-1,1)$, $\Phi^{-1}(.)$ the quantile of the univariate standard normal distribution. Here, we expect method (c) to fail always, since the Gaussian copula is not even an extreme-value copula.
\end{itemize}
We set the global parameters to $n = 1\,000\,000$ samples for methods (b) and (c) and the exponential rates of the marginals $\lambda_1=\lambda_2=0.1$. The parameters of the copulas are $\alpha=2/3,\,\theta=0.5$, and $\rho = 1/\sqrt{2}$, implying that all three copulas have a Kendall's Tau of $0.5$, which follows from the formulas in \cite[Example 1.11]{maibook12}, \cite[Example 6.7]{embrechts01}, and \cite[p.\ 215 ff]{embrechts05}, respectively. The results of our simulation are provided in Table \ref{table_case_study}. It can be observed from Table \ref{table_case_study} that in the case $T=S$ Algorithm \ref{algoSimu} is exact\footnote{We call a simulation result ``exact'' if the relative error of our empirical estimator, based on the $n=1\,000\,000$ simulations, is smaller than $0.5 \%$.} for the Gumbel copula and for the Marshall--Olkin distribution. The reason for this is that both underlying distributions are min-stable multi-variate exponential distributions, implying that $\min\{\tau_1,\tau_2\}$ is exponential, and therefore the one-dimensional indicator process $\I{\{\tau_1>t,\tau_2>t\}}$ is Markovian. However, if $S \neq T$ Algorithm \ref{algoSimu} is only exact for the Marshall--Olkin distribution, because the bivariate survival indicator process is not Markovian in the Gumbel case. As expected, for the Gaussian copula Algorithm \ref{algoSimu} is strongly biased, because it simply is wrong.
\par
\begin{table}
	\centering
		\begin{tabular}{llll}
		copula &exact value (a)&method (b)& method (c)\\
		\hline 
		$S=T=10$\\
\hline
		Marshall--Olkin	&0.26360 &0.26356 (0.013 \%)& 0.26288 (0.271 \%)\\
		Gumbel 					&0.24312  &0.24275 (0.150 \%)& 0.24238 (0.302 \%)\\
		Gaussian 				&0.14542 &0.14521 (0.151 \%)&\textbf{0.14309} (\textbf{1.604} \%)\\
		\hline
		$S=T/2=5$\\
\hline
	  Marshall--Olkin	&0.31140 &0.31215 (0.241 \%)&0.31154 (0.043 \%)\\
		Gumbel 					&0.32692 &0.32696 (0.010 \%)&\textbf{0.29916} (\textbf{8.491} \%)\\
		Gaussian 				&0.32908 &0.32932 (0.073 \%)&\textbf{0.29504} (\textbf{10.344} \%)\\
		\hline
		\end{tabular}
	\caption{Results of our case study, with relative errors with respect to method (a) in parentheses. The numerical results of the cases where the respective simulation approach implies a bias are displayed in bold.}
	\label{table_case_study}
\end{table}

\section{Conclusion}
The industry practice of economic scenario generation, involving dependent default times, is critically reviewed. As a possible trade-off between realistic stylized facts, practical demands, and mathematical viability, the class of default models with a Markovian survival indicator process is discussed. The ``looping default'' model, an example from this class, is linked to the classical ``Freund distribution'' and a new construction (with immediate multi-variate extensions) based on Markov chains is given. If additionally all sub-vectors of the survival indicator process are Markovian, this constitutes a new characterization of the Marshall--Olkin law. En passant, this shows that the model features ``consistent marginalization'' and ``nested sub-distributions'' that are still of Marshall--Olkin type, connecting multi-variate lack-of-memory with consistent nested marginalization. For this model, we present an efficient and unbiased simulation scheme based on a multi-factor L\'evy-frailty construction. Throughout the paper and within a numerical case study, we work out different pitfalls of simulating dependent default times, giving a word of caution on inadequate approaches that are still used in the financial industry. It is simply not possible to iterate default time simulation while preserving the dependence structure, unless the default times are jointly distributed with a Marshall--Olkin law and, in particular, it is wrong to iterate a Gaussian copula on the default times while assuming the final overall default monitoring to be consistent with a one-period Gaussian copula.

%
%
%
%
%
%

\end{document}